\documentclass[letterpaper,10pt,conference]{ieeeconf} 

\IEEEoverridecommandlockouts         
\usepackage{graphicx}                
\usepackage{mathptmx}                

\usepackage{amsmath,amssymb,mathtools} 
\usepackage{cancel}                  
\usepackage{physics}                 

\usepackage{tikz}
\usetikzlibrary{decorations.pathmorphing}
\usepackage{circuitikz} 

\usepackage{makecell}                
\usepackage{comment}                 

\let\labelindent\relax               
\usepackage{enumitem}                

\usepackage{amsthm}
\theoremstyle{plain}
\newtheorem{theorem}{Theorem}

\makeatletter
\let\NAT@parse\undefined
\makeatother
\usepackage[numbers,sort&compress]{natbib}

\usepackage{xcolor}
\definecolor{linkblue}{RGB}{0,0,150}

\usepackage[
  colorlinks=true,
  linkcolor=red,
  citecolor=blue,
  urlcolor=linkblue
]{hyperref}

\title{\LARGE \bf
Small HVAC Control Demonstrations in Larger Buildings Often Overestimate Savings
}

\author{Arash J. Khabbazi and Kevin J. Kircher%
\thanks{The authors are with the School of Mechanical Engineering,
Purdue University, West Lafayette, IN 47907, USA. 
{\tt\small \{arashjkh,kkirche\}@purdue.edu}}}

\begin{document}

\maketitle
\thispagestyle{empty}
\pagestyle{empty}

\begin{abstract}
How much energy, money, and emissions can advanced control of heating and cooling equipment save in real buildings? To address this question, researchers sometimes control a small number of thermal zones within a larger multi-zone building, then report savings for the controlled zones only. That approach can overestimate savings by neglecting heat transfer between controlled zones and adjacent zones. This paper mathematically characterizes the overestimation error when the dynamics are linear and the objectives are linear in the thermal load, as usually holds when optimizing energy efficiency, energy costs, or emissions. Overestimation errors can be large even in seemingly innocuous situations. For example, when controlling only interior zones that have no direct thermal contact with the outdoors, all perceived savings are fictitious. This paper provides an alternative estimation method based on the controlled and adjacent zones' temperature measurements. The new method does not require estimating how much energy the building would have used under baseline operations, so it removes the additional measurement and verification challenge of accurate baseline estimation.
\end{abstract}

\section{Introduction}

Residential and commercial heating, ventilation, and air conditioning (HVAC) systems give rise to about 15\% of anthropogenic greenhouse gas emissions and to energy costs on the order of \$1 trillion per year \cite{Khabbazi-2025-Field-HVAC-MPC-RL}. A large body of simulation research suggests that advanced HVAC control methods, such as model predictive control (MPC; see \cite{Serale2018-ph}, \cite{Drgona2020-fh}) and reinforcement learning (RL; see \cite{Vazquez-Canteli2019-ey}, \cite{Wang2020-lx}), could reduce energy costs, emissions, and peaks in electricity demand. However, advanced HVAC control systems have not seen much real-world adoption \cite{Henze-etal-2024}. This is partly because decision-makers in business and government may doubt that advanced control systems can reliably deliver persistent real-world savings on par with simulation results. Field demonstrations can help convince decision-makers in business and government that advanced HVAC control can actually deliver significant benefits in the real world \cite{Khabbazi-etal-2024-Herrick}. Researchers to date have done a relatively small number of field demonstrations, highlighting the need for further studies \cite{Pergantis-etal-2024a}. Field demonstrations must also be trustworthy, which requires reliable measurement and verification (M\&V).

Unfortunately, a recent review covering field demonstrations of advanced HVAC control revealed a systematic problem with M\&V \cite{Khabbazi-2025-Field-HVAC-MPC-RL}. Specifically, researchers sometimes control a small number of thermal zones within a larger multi-zone building, then report energy savings for the controlled zones only. That approach risks overestimating performance improvements by neglecting heat transfer between controlled zones and adjacent zones. The simple example of heating one conference room within a large office building illustrates this M\&V problem. If an advanced control system reduces the temperature in the conference room (to save energy during an unoccupied period, for example), then heat will flow into the conference room from adjacent zones. Maintaining the default temperatures in the adjacent zones will then require more energy to compensate for the additional heat loss to the conference room. Failure to account for the additional energy used in adjacent zones will lead to overestimating the net benefits of advanced control. Of the 80 total field demonstration papers in commercial buildings reviewed in \cite{Khabbazi-2025-Field-HVAC-MPC-RL}, only 34 (43\%) controlled a whole building. Of the remaining 46 papers (57\%) that did not control a whole building, no paper explicitly accounted for heat transfer with adjacent zones when reporting energy or cost savings.

The traces of this M\&V problem date back to the early 1990s, when Ruud et al. \cite{Ruud1990} tried to reduce thermal interactions between their controlled test floor and adjacent floors by pre-cooling the adjacent floors. However, Ruud et al. found strong convective flows from warmer zones, which contributed to discrepancies in measured cooling energy. In 2002, Braun et al. \cite{Braun2002} found that cross-zone heat transfer significantly influenced estimates of energy savings. They addressed this by expanding their control system to include the whole building. In 2003, Braun \cite{Braun2003} noted that thermal coupling between controlled and adjacent zones reduced the effectiveness of pre-cooling strategies and required careful experimental design to prevent misleading outcomes. In 2005, Henze et al. \cite{Henze2005} also highlighted this problem when field-testing their MPC system. They controlled all test rooms uniformly and removed the general area from performance evaluation due to heat transfer with the controlled spaces. Early research on RL by Liu and Henze \cite{Liu2006} also acknowledged significant thermal coupling between controlled zones and adjacent uncontrolled zones. 

While this M\&V issue was noticed and partially corrected for in early field demonstrations, a wide range of more recent field demonstrations have reverted to neglecting cross-zone heat transfer and potentially overestimating savings. This reversion could be because no paper has carefully explained the underlying cause of the issue and how to avoid it. This paper aims to fill that gap by mathematically characterizing the causes and extent of overestimation errors when controlling a small number of thermal zones within a larger multi-zone building. This paper also develops a new method to eliminate overestimation errors by accounting for cross-zone heat transfer. The new method requires only zone air temperature measurements, which are readily available in most buildings, and estimates of thermal conductances between zones. The new method does not require estimating an energy baseline, so it removes a recurring M\&V challenge.

Section~\ref{theorem} of this paper describes the assumptions and thermal model used for analysis. It then presents a theorem that characterizes the overestimation error and proposes a corrected savings estimation method. Section~\ref{numerical-example} presents an illustrative numerical example. Section~\ref{limitations} discusses limitations and directions for future work. Section~\ref{conclusion} concludes the paper. The appendix proves Theorem~\ref{theorem-1}. 

\section{Main Theoretical Results}
\label{theorem}

\subsection{Assumptions and Model}
\label{model} 

Thermal circuits model heat transfer by analogy to electrical circuits. In this analogy, heat plays the role of charge and temperature plays the role of voltage. Thermal resistances and capacitances replace their electrical analogues. A thermal capacitance can represent the air inside a room or a group of adjacent rooms; or it can represent the thermal mass associated with building materials, furnishings, {\it etc.} Heat transfers through network connections that represent conduction, convection, or radiation. Thermal resistances represent things that impede heat transfer, such as wall insulation.

This paper considers an arbitrary thermal network topology with an arbitrary number $n$ of thermal zones. The positive integer $n$ could be very large if, for example, the thermal network represents a high-resolution spatial discretization of the partial differential equations that govern heat transfer through continuous media in multiple dimensions. Fig.~\ref{thermalCircuits} shows a generic zone $i$ within the arbitrary $n$-zone building model considered here. In this model, the temperature in each thermal zone $i = 1$, \dots, $n$ satisfies
\begin{equation}
\begin{aligned}
C_i \frac{\text d T_i(t)}{\text d t} &= \sum_{j=0}^n \alpha_{ij} ( T_j(t) - T_i(t) ) + q_i(t) + w_i(t) . \label{baselineDynamics} \\
\end{aligned}
\end{equation}
Here $t \in [0,\tau]$ (h) denotes time, $\tau$ (h) is the final time, $T_i(t)$ ($^\circ$C) is the temperature of zone $i$ (with the outdoor air viewed as `zone zero'), $\alpha_{ij}$ (kW/$^\circ$C) is the thermal conductance (the inverse of the thermal resistance $R_{ij}$ [$^\circ$C/kW]) between zones $i$ and $j$, $q_i(t)$ (kW) is the thermal power delivered to zone $i$ by heating or cooling equipment, and $w_i(t)$ (kW) is the thermal power delivered to zone $i$ by exogenous sources such as lights, electronics, body heat, and the sun. The controlled thermal power $q_i(t)$ is positive for heating and negative for cooling. If no equipment supplies heat directly to (or removes heat directly from) zone $i$, then $q_i(t) = 0$ for all $t$. In buildings with significant thermal mass, for example, the thermal mass associated with a wall or floor can be represented as a zone with $q_i(t) \equiv 0$. If zone $i$ has no direct thermal contact with zone $j$, then $\alpha_{ij} = 0$. By symmetry, $\alpha_{ij} = \alpha_{ji}$ for all $(i,j)$.

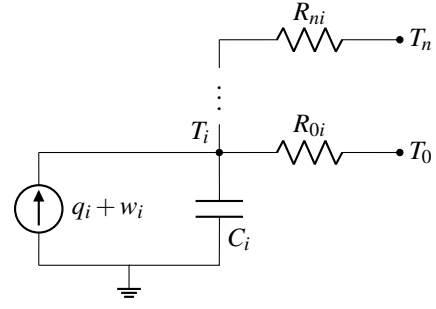
\begin{figure}[t]
\begin{center}
\begin{circuitikz}[xscale=1.2, yscale=0.75, american currents] 
\ctikzset{bipoles/length=1.05cm}

    \pgfmathsetmacro{\w}{2};
    \pgfmathsetmacro{\h}{1};

    \node[above left] at (2*\w,2*\h) {$T_i$};
    \draw (\w, 2*\h) -- (2*\w,2*\h) to[C,*-] (2*\w,0) -- (\w,0);
    \node[above right] at (2*\w,0.1) {$C_i$};

    \draw (\w,0) to[I,n=I] (\w,2*\h);
    \node[right] at (I.s) {$q_i + w_i$};
    \draw (1.5*\w,0) node[ground] {} to (1.5*\w,0);

    \draw (2*\w,2*\h) to[R,R=$R_{0i}$,-*] (3*\w,2*\h);
    \node[right] at (3*\w,2*\h) {$T_0$};
    
    \draw (2*\w,2*\h) -- (2*\w,2.5*\h);
    \node[rotate=90] at (2*\w,3*\h) {\dots};
    \draw (2*\w,3.5*\h) -- (2*\w,4*\h);
    
    \draw (2*\w,4*\h) to[R,R=$R_{ni}$,-*] (3*\w,4*\h);
    \node[right] at (3*\w,4*\h) {$T_n$};

\end{circuitikz}
\end{center}
\caption{Zone $i$ in a thermal circuit representing an arbitrary $n$-zone building.}
\label{thermalCircuits}
\end{figure}

This paper considers two operating scenarios for the building. In the {\it baseline} scenario, a benchmark control system operates all zones. In the {\it experiment} scenario, an advanced control system takes over operation of only $m$ zones (labeled zones $1$, \dots, $m$ without loss of generality) for some positive integer $m \leq n$, while the other $n - m$ zones (labeled zones $m+1$, \dots, $n$) remain under benchmark control. This setup includes the limiting cases where only $m = 1$ zone is controlled, where all $m = n$ zones are controlled, and all cases in between. Aside from using different controllers in zones $1$, \dots, $m$, the baseline and experiment scenarios are identical. Both scenarios have the same building physics, initial and final states, weather conditions, and occupant behavior. Therefore, the experiment scenario has temperature dynamics
\begin{equation}
\begin{aligned}
C_i \frac{\text d \tilde T_i(t)}{\text d t} &= \sum_{j=0}^n \alpha_{ij} ( \tilde T_j(t) - \tilde T_i(t) ) + \tilde q_i(t) + w_i(t) , \label{experimentDynamics} 
\end{aligned}
\end{equation}
where a variable with a tilde, such as $\tilde T_i(t)$, denotes its value in the experiment scenario.   In the experiment scenario, the advanced control system changes the temperatures $T_1$, \dots, $T_m$ and controlled thermal powers $q_1$, \dots, $q_m$, so $\tilde T_i \neq T_i$ and $\tilde q_i \neq q_i$ for zones $i = 1$, \dots, $m$. Even though every other zone $i = m+1$, \dots, $n$ remains under baseline control at the unperturbed temperature $\tilde T_i = T_i$, the thermal power  $q_i$ in general changes to some perturbed value  $\tilde q_i \neq q_i$ due to altered heat transfer between zone $i$ and zones $1$, \dots, $m$. 

The instantaneous cost associated with delivering thermal power $q_i(t)$ to zone $i$ in the benchmark scenario is $c_i(t) = a_i(t) q_i(t) + b_i(t)$. Similarly, delivering thermal power $\tilde q_i(t)$ to zone $i$ in the experiment scenario costs $\tilde c_i(t) = a_i(t) \tilde q_i(t) + b_i(t)$. The aggregate (over zones) cumulative (over time) cost savings with respect to the benchmark scenario are therefore 
\[
\begin{aligned}
S &\coloneqq \sum_{i=1}^n \int_0^\tau (c_i(t) - \tilde c_i(t)) \text d t = \sum_{i=1}^n \int_0^\tau a_i(t) (q_i(t) - \tilde q_i(t)) \text d t .
\end{aligned}
\]
Depending on the definition of the thermal price $a_i(t)$, the cost $\int_0^\tau c_i(t) \text d t$ can represent a variety of objectives, such as reducing heat or cooling demand, input fuel or electricity, energy costs, emissions, or combinations of these.

\subsection{Theorem and Remarks}
\begin{theorem}
\label{theorem-1}
If the assumptions and model in \S~\ref{model} hold, then reporting only the observed cost savings for the controlled zones 1, \dots, $m$ overestimates the whole-building cost savings by 
\noindent\begingroup\small
\[
\sum_{i=1}^m \sum_{j=m+1}^n \alpha_{ij}  \int_0^\tau a_j(t) (T_i(t) - \tilde T_i(t)) \text d t .
\] \endgroup
To correct this overestimation error, researchers can report either
\noindent\begingroup\small
\[
\sum_{i=1}^m \int_0^\tau \left[ a_i(t) \alpha_{i0} + \sum_{j=1}^n \alpha_{ij} (a_i(t) - a_j(t) ) + C_i a_i(t)  \frac{\text d }{\text d t} \right] (T_i(t) - \tilde T_i(t) ) \text d t 
\] \endgroup
or
\noindent\begingroup\small
\[
\sum_{i=1}^m \int_0^\tau \left[  a_i(t) \alpha_{i0} + \sum_{j=1}^n \alpha_{ij} (a_i(t) - a_j(t) ) - C_i \frac{\text d a_i(t) }{\text d t} \right] (T_i(t) - \tilde T_i(t)) \text d t ,
\] \endgroup
which are two equivalent and exact expressions for the whole-building cost savings.
\end{theorem}

A simple two-zone example illustrates the implications and implementation details of Theorem~\ref{theorem-1}. If the advanced control system controls only zone $m = 1$ in a building with $n = 2$ zones, and if the thermal prices are spatially uniform, then the savings overestimation error is
\[
\alpha_{12} \int_0^\tau a_1(t) (T_1(t) - \tilde T_1(t)) \text d t 
\]
and the true savings are
\[
\int_0^\tau \left(  a_1(t) \alpha_{10} - C_1 \frac{\text d a_1(t) }{\text d t} \right) (T_1(t) - \tilde T_1(t)) \text d t .
\]
In the limit of slowly-changing prices, d$a_1(t)/$d$t \rightarrow 0$ and the true savings simplify to
\[
\alpha_{10}  \int_0^\tau a_1(t) (T_1(t) - \tilde T_1(t)) \text d t .
\]
In this limit, the relative overestimation error,
\[
e \coloneqq \frac{ \text{overestimation error} }{ \text{true savings} } = \frac{ \alpha_{12} }{ \alpha_{10} } ,
\]
is independent of the temperatures and prices. In this limit, the estimation error increases if the thermal coupling strength $\alpha_{12}$ between the controlled zone and the adjacent uncontrolled zone increases; or if the coupling strength $\alpha_{10}$ between the controlled zone and the outdoors decreases. In the limit $\alpha_{10} \rightarrow 0$, meaning there is no direct coupling between the controlled zone and the outdoors, the overestimation error goes to infinity. This implies that when controlling only a purely interior zone, all perceived savings are fictitious.

The thermal conductance between zones 1 and 2 can be written as $\alpha_{12}=U_{12}A_{12}$, where $U_{12}$ (kW/[$^\circ$C\,m$^2$]) is the overall heat transfer coefficient of the surface separating Zone~1 from Zone~2, and $A_{12}$ (m$^2$) is the surface area. Letting $\beta = U_{12} / U_{10}$ and $\gamma = A_{12}/A_{10}$, the relative error is
\[
e = \frac{U_{12} A_{12}}{U_{10} A_{10}} =  \beta \gamma .
\]
For example, a box-shaped zone with a square footprint, adiabatic floor and ceiling, and $\ell$ exterior walls has $4-\ell$ interior walls, so $\gamma = 4/\ell - 1$ and $e = (4/\ell - 1) \beta$. If Zone~2 entirely surrounds Zone~1, meaning the controlled zone is an interior zone, then $\ell = 0$, $e = \infty$, and advanced control achieves no savings. At the other extreme, if Zone~1 shares no walls with Zone~2, meaning the two zones are effectively separate buildings, then $\ell = 4$, $e = 0$, and savings estimates are exact. Between these two extremes, the relative error scales like $1/\ell$. With $\ell = 2$ exterior walls, for example, the relative error is $e = \beta$. If the exterior walls have double the insulation of the interior walls, then $\beta = 2$ and the relative error is 200\%. In other words, the true savings are one-third of the reported savings. Fig.~\ref{square} shows how the relative savings overestimation error scales with $\gamma$ and $\beta$.

\begin{figure}[t]
\begin{center}
\includegraphics[width=0.39\textwidth]{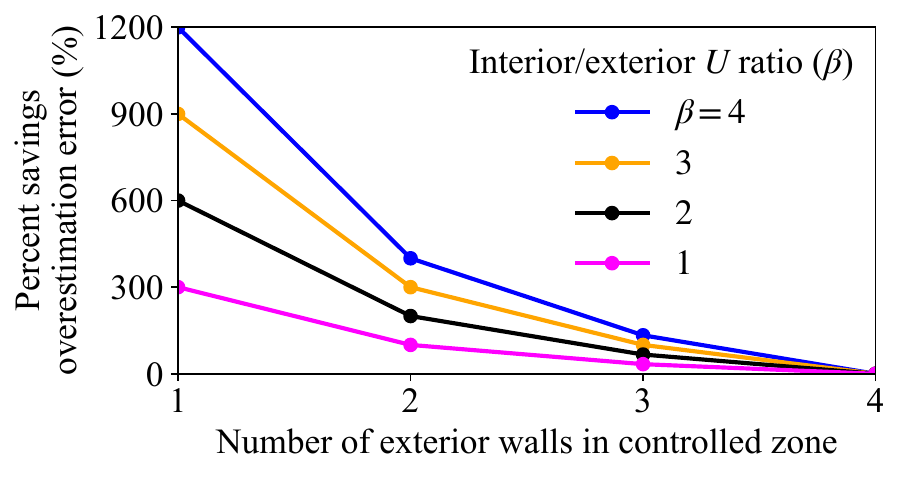}
\end{center}
\caption{Savings overestimation errors for a square zone with adiabatic floor and ceiling. Errors are larger for controlled zones with less exterior surface area, more exterior insulation, or less interior insulation.}
\label{square}
\end{figure}

\section{Numerical Example}
\label{numerical-example}

\subsection{Setting, Building, and Input Data}
This section illustrates the themes in this paper through the example of heating a two-zone building in West Lafayette, Indiana, USA, during a five-day cold snap in late December of 2022. Fig.~\ref{floorPlan} shows the example building's floor plan. The advanced control system operates only Zone~1, which comprises one-quarter of the building's floor area. The default control system maintains baseline operations in the other three-quarters of the building (Zone~2).

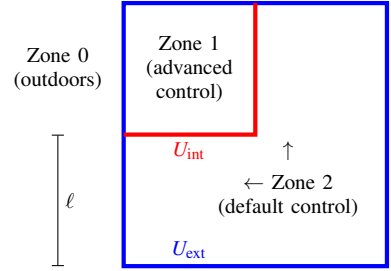
\begin{figure}[t]
\begin{center}
\footnotesize
\begin{tikzpicture}[scale=1.75]	

	\draw[ultra thick,blue] (0,0) rectangle (2,2); 
	\draw[ultra thick,red] (0,1) -- (1,1) -- (1,2); 
	
	\node[align=center] at (0.5,1.5) {Zone~1 \\ (advanced \\ control)};
	\node[above, align=center] at (1.25,0.75) {$\uparrow$};
	\node[below, align=center] at (1.25,0.75) {$\leftarrow$ Zone~2 \\ (default control)};
	\node[align=center] at (-0.5,1.5) {Zone~0 \\ (outdoors)};
	
	\node[above,blue] at (0.5,0) {$U_\text{ext}$};
	\node[below,red] at (0.5,1) {$U_\text{int}$};
	
	\draw[|-|] (-0.5,0) -- (-0.5,1);
	\node[right] at (-0.5,0.5) {$\ell$};
			
\end{tikzpicture}
\end{center}
\caption{Advanced control system operates only Zone~1 within the two-zone building. {\color{blue} Exterior} walls have better insulation than {\color{red} interior} walls.}
\label{floorPlan}
\end{figure}

The floors and ceilings are assumed to be perfectly insulated. The walls are $h = 3$~m tall and $\ell =$ 5~m long, so Zone~1 has exterior and interior wall areas $A_{10} = A_{12} = 2 \ell h = 30$~m$^2$. Zone~2 has exterior and interior wall areas $A_{20} = 6 \ell h =$ 90~m$^2$ and $A_{21} = A_{12}$. The interior and exterior walls have overall heat transfer coefficients $U_\text{int} = 3$~W/($^\circ$C\,m$^2$) and $U_\text{ext} = 1.5$~W/($^\circ$C\,m$^2$). The thermal conductances are therefore $\alpha_{10} = A_{10} U_\text{ext} = 45$~W/$^\circ$C, $\alpha_{12} = \alpha_{21} = A_{12} U_\text{int} =$~90 W/$^\circ$C, and $\alpha_{20} = A_{20} U_\text{ext} = 135$~W/$^\circ$C. 

The thermal capacitance of Zone $i$ comes from multiplying the density $\rho = 1.293$~kg/m$^3$, specific heat at constant pressure $c_p = 2.792 \times 10^{-4}$~kWh/($^\circ$C\,kg), and volume $V_i$ of air in the zone, then increasing the result by a factor of ten to account for ducts and other matter in tight thermal contact with the air. With the zone volumes $V_1 = \ell^2 h = 75$~m$^3$ and $V_2 = 3 \ell^2 h =$ 225~m$^3$, this process produces thermal capacitances $C_1 = 0.27$~kWh/$^\circ$C and $C_2 = 0.81$~kWh/$^\circ$C. 

In the baseline scenario, the benchmark control system holds the zones at constant temperatures $T_1$ and $T_2$, so solving the dynamics~\eqref{baselineDynamics} for the controlled thermal powers gives
\[
\begin{aligned}
q_1(t)  &= \alpha_{10} ( T_1 - T_0(t) ) + \alpha_{12} ( T_1 - T_2 ) - w_1(t) \\
q_2(t)  &= \alpha_{20} ( T_2 - T_0(t) ) + \alpha_{12} ( T_2 - T_1 ) - w_2(t) . \\
\end{aligned}
\]
In the experiment scenario, the benchmark control system maintains Zone~2 at constant temperature $T_2$, but the advanced control system perturbs the Zone~1 temperature to some dynamic value $\tilde T_1(t) \neq T_1$. Therefore, the dynamics~\eqref{experimentDynamics} reduce to
\noindent\begingroup\small
\[
\begin{aligned}
C_1 \frac{\text d \tilde T_1(t)}{\text d t} &= \alpha_{10} ( T_0(t) - \tilde T_1(t) ) + \alpha_{12} ( T_2 - \tilde T_1(t) ) + \tilde q_1(t) + w_1(t) \\
\tilde q_2(t) &= \alpha_{20} ( T_2 - T_0(t) ) + \alpha_{12} ( T_2 - \tilde T_1(t) ) - w_2(t) . \\
\end{aligned}
\] \endgroup
Assuming $T_0$, $\tilde q_1$, and $w_1$ are piecewise constant over discrete time steps of duration $\Delta t$ (h), the Zone~1 dynamics discretize exactly to
\noindent\begingroup\small
\begin{align}
&\tilde T_1(k+1) = \lambda \tilde T_1(k)  \nonumber \\
&+ ( 1 - \lambda) \big( \tilde q_1(k) + \alpha_{10} T_0(k) + \alpha_{12} T_2 + w_1(k) \big) / (\alpha_{10} + \alpha_{12} ) ,
\label{exampleDynamics}
\end{align} \endgroup
where $\lambda = \exp( -\Delta t (\alpha_{10} + \alpha_{12}) / C_1 )$ and the integer $k$ indexes time steps.

The exogenous thermal powers $w_1(t)$ and $w_2(t)$ come from the sun, electronics, body heat, and other sources. The simulations that follow use solar heat gains obtained by rescaling the global solar irradiance on a horizontal surface (available from historical weather measurements) to a mean of 10~W/m$^2$, then multiplying by the window area, assuming a window-to-wall ratio of 25\%. This results in peak solar heat gains of 0.47 and 1.4~kW for Zones 1 and 2. The solar heat gains combine with other internal heat gains averaging 10~W/m$^2$, normalized by floor area and perturbed with white Gaussian noise, to form $w_1(t)$ and $w_2(t)$.

\begin{figure}[t]
\begin{center}
\includegraphics[width=0.49\textwidth]{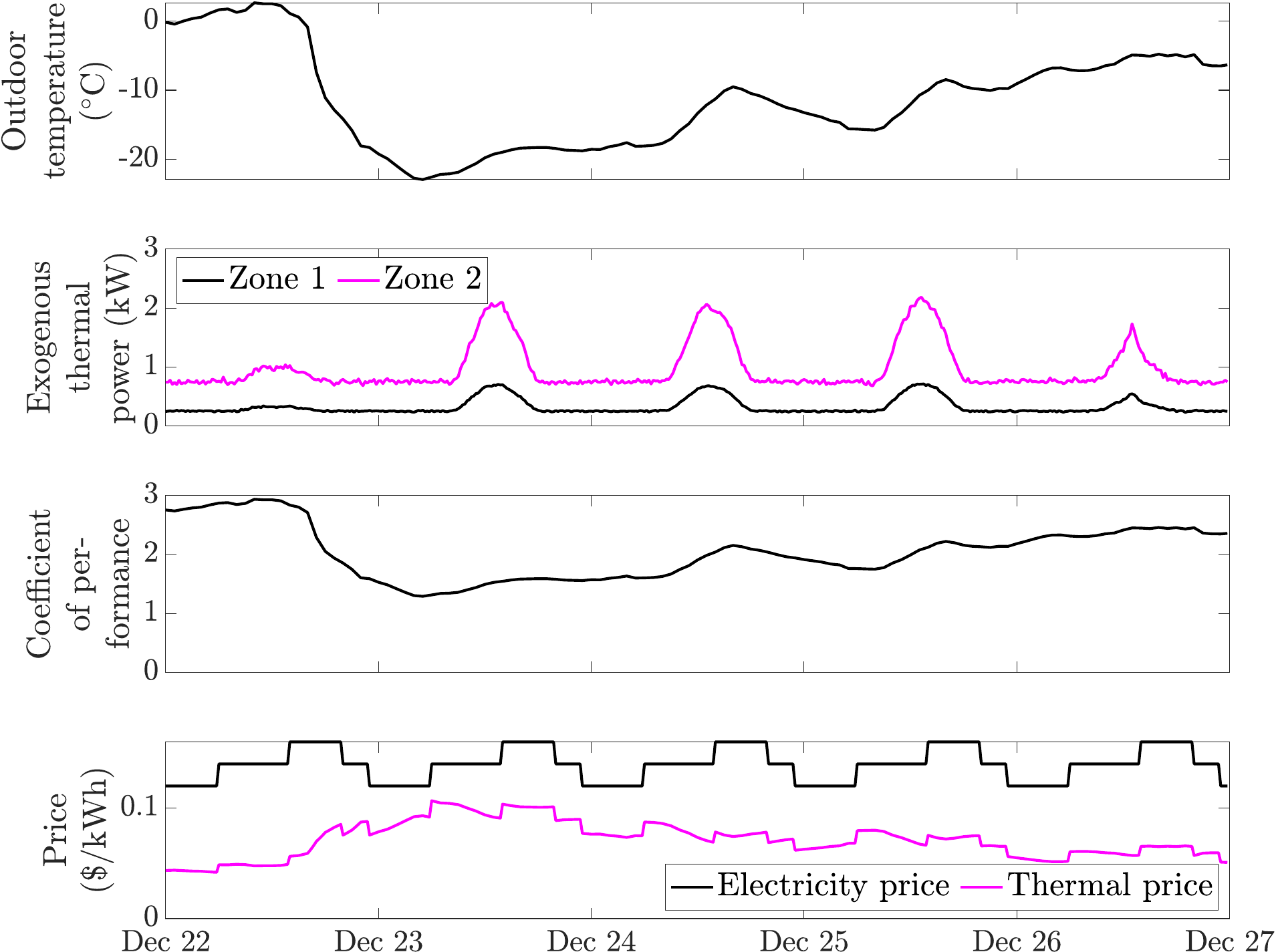}
\end{center}
\caption{Simulation input data represent five cold days in West Lafayette, Indiana, from December of 2022.}
\label{inputs}
\end{figure}

The advanced control system seeks to minimize the cost of input energy to the electric heat pump that heats the building. The electricity price $\pi(t)$ varies with time following a rate plan from a large utility in Chicago, Illinois. The price is 0.12~\$/kWh from 10 PM to 6 AM, 0.14~\$/kWh from 6 AM to 2 PM and 7 to 10 PM, and 0.16 \$/kWh from 2 to 7 PM. The heat pump's coefficient of performance $\eta(t)$ also varies with time due to its dependence on the outdoor temperature. The coefficient of performance scales linearly between 1.8 at -15~$^\circ$C and 3.3 at 8.3~$^\circ$C. The thermal price $a(t) = \pi(t) / \eta(t)$ is the same for Zones 1 and 2. The optimization problem is solved in MATLAB using the CVX toolbox \cite{Grant2008CVX}.

Fig.~\ref{inputs} shows the input data for five cold days in December 2022. The outdoor air temperature (top plot) drops as low as -23~$^\circ$C. The exogenous thermal powers (second plot) are larger for Zone~2 (black curve), which has triple the floor area of Zone~1 (magenta curve). The heat pump's coefficient of performance (third plot) drops in cold weather. The electricity price (black curve) follows the time-of-day rate plan. The effective thermal price (magenta curve) is the ratio of the electricity price to the heat pump coefficient of performance.

In the simulations that follow, the advanced control system solves a convex optimization problem for the decision variables $\tilde T_1(0)$, \dots, $\tilde T_1(K)$ and $\tilde q_1(0)$, \dots, $\tilde q_1(K-1)$. The objective function is the cumulative electricity cost,
\[
\Delta t \sum_{k=0}^{K-1} a(k) \tilde q_1(k) = \Delta t \sum_{k=0}^{K-1} \frac{\pi(k) \tilde q_1(k)}{\eta(k)} .
\]
The dynamics~\eqref{exampleDynamics} enter the optimization problem as an equality constraint for each $k = 0$, \dots, $K-1$. The optimization problem also imposes the equality constraints 
\[
\tilde T_1(0) = \tilde T_1(K) = T_1
\]
to ensure that Zone~1 begins and ends in the same state in the simulated experiment as it would under baseline control. The thermal powers $\tilde q_1(0)$, \dots, $\tilde q_1(K-1)$ must be nonnegative. The Zone~1 air temperature $\tilde T_1(k)$ must deviate no further than $\delta(k)$ from the baseline temperature,
\[
\abs{ \tilde T_1(k) - T_1} \leq \delta(k) .
\]
The optimization permits temperature deviations of $\delta(k) = 1$ $^\circ$C while occupants are at home and awake (6 to 9 AM and 6 to 10 PM) and of $\delta(k) = 2$~$^\circ$C otherwise.

\subsection{Simulation Results}

Fig.~\ref{results} shows the simulation results. The left column of plots shows results for Zone~1 and the right column for Zone~2. The top row of plots shows the indoor temperatures under the baseline (black curves) and experiment (magenta curves) scenarios, as well as the temperature constraints (dashed red curves). The Zone~2 temperatures are identical under the baseline and experiment scenarios. In Zone~1, however, the advanced control system perturbs the temperature away from the baseline during the experiment scenario. The Zone~1 temperature drops to the minimum acceptable value at most time steps, but preheats occasionally in advance of large step-increases of the thermal price $a(k)$ (the magenta curve in the bottom plot of Fig.~\ref{inputs}).

Lowering the Zone~1 temperature at most time steps reduces the heat demand of Zone~1 in the experiment scenario (magenta curve, middle left plot) relative to the baseline scenario (black curve). This reduces the cost of input electricity to the heat pump (magenta curve, lower left plot) relative to the baseline (black curve). Over the five days, heating Zone~1 costs \$8.34 in the experiment scenario and \$10.17 in the baseline scenario. Looking only at Zone~1 would suggest savings of \$1.83 (18\%).

However, lowering the Zone~1 temperature also causes heat to flow from Zone~2 to Zone~1. This increases the heat demand of Zone~2 (magenta curve, middle right plot) relative to the baseline scenario (black curve). This, in turn, increases the cost of input electricity to the heat pump in the experiment scenario (magenta curve, lower right plot) relative to the baseline scenario (black curve). Over the five days, heating Zone~2 costs \$31.70 in the experiment scenario and \$30.49 in the baseline scenario. Advanced control of Zone~1 {\it increases} the cost of heating Zone~2 by \$1.21 (4\%). 

Accounting for both the \$1.83 cost savings in Zone~1 and the \$1.21 cost increase in Zone~2, the total cost of heating the whole building is \$40.66 in the baseline scenario and \$40.04 in the experiment scenario. The true whole-building savings are therefore \$0.62. The \$1.83 savings suggested by looking only at Zone~1 overestimate the true \$0.62 whole-building savings by \$1.21 (195\%). This agrees closely with the overestimation predicted by the theorem, $\alpha_{12} / \alpha_{10} = 200$\%. Table~\ref{resultsTable} summarizes these results.

\begin{figure}[t]
\begin{center}
\includegraphics[width=0.49\textwidth]{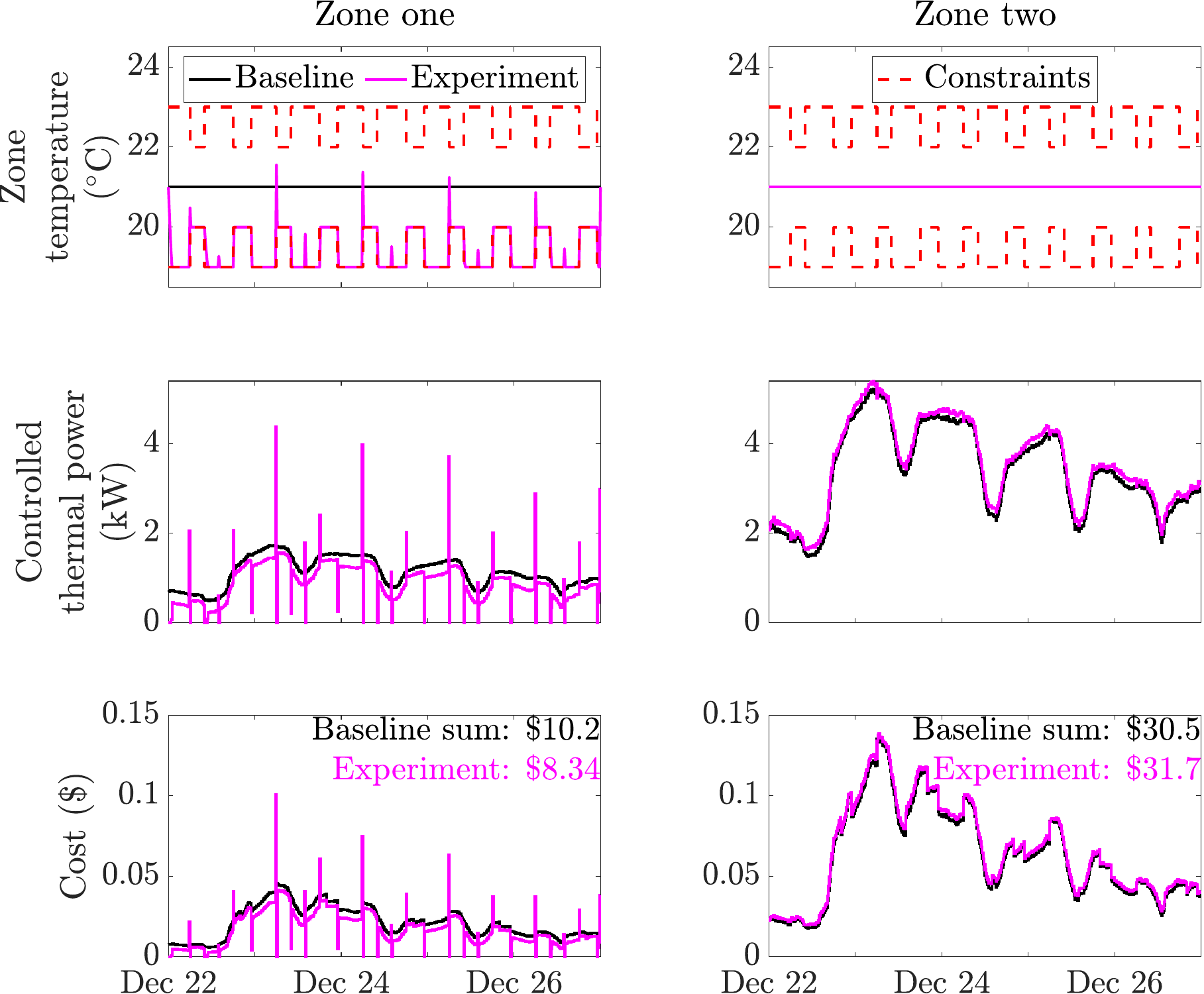}
\end{center}
\caption{Zone temperatures (top row), heat pump power output (middle), and electricity cost (bottom) for Zone~1 under advanced control (left column) and for Zone~2 under default control (right column). While advanced control decreases energy costs in Zone~1, it also increases energy costs in Zone~2.}
\label{results}
\end{figure}

\section{Limitations and Future Work}
\label{limitations}

The scope of this paper only included settings where heat transfer is linear in temperature differences and control objectives are linear in thermal load. While heat transfer is not always linear in temperature differences -- for example, radiative heat transfer follows differences in the fourth powers of temperatures -- it is approximately linear in many applications of practical interest. Many important control objectives are also linear or approximately linear in thermal load, such as input energy to HVAC equipment, the cost of input energy, and emissions associated with input energy. However, some important objectives, such as reducing peak electricity demand, are not linear in thermal load. When the cost does not have the form $c_i(t) = a_i(t) q_i(t) + b_i(t)$, the savings $S$ are not linear in the thermal loads $q_i(t)$ and the proof of Theorem \ref{theorem-1} does not hold. Future work could investigate nonlinear objectives such as peak shaving.

The proposed M\&V procedure requires the thermal conductances $\alpha_{ij}$, which can be challenging to estimate accurately. An important direction for future work is to investigate how uncertainty in the $\alpha_{ij}$ affects estimation of the true savings, and to develop mitigation strategies that reduce sensitivity to this uncertainty. Another future direction could compare zone-level savings estimates to `true' savings in the optimistic benchmark of applying advanced control to the whole building. Ultimately, the theory in this paper should be validated through field experiments. A simple demonstration case could be a two-story house where advanced control is applied only to one floor, while the adjacent floor remains under baseline control.

\begin{table}[t]
\caption{Energy costs and savings in the numerical example}
\begin{center}
\renewcommand{\arraystretch}{1.1}

\begin{tabular}{ l | l | l | l }
& Zone~1 & Zone~2 & Whole building \\
\hline
Baseline cost & \$10.17 & \$30.49 & \$40.66 \\
Experiment cost & \$8.34 & \$31.70 & \$40.04 \\
Perceived zone-level savings & {\color{red} \$1.83} & $-\$1.21$ & -- \\
True savings & -- & -- & {\color{blue} \$0.62 } \\

\end{tabular}
\end{center}
\label{resultsTable}
\end{table}

\section{Conclusion}
\label{conclusion}
This paper highlighted a problem with measuring and verifying savings from small HVAC control demonstrations in larger multi-zone buildings. The problem involves heat transfer between controlled zones and adjacent zones that, if neglected, can lead to large savings overestimation errors. This paper proposed a new measurement and verification procedure that accounts for heat transfer with adjacent zones. Surprisingly, the new procedure does not require estimating how much energy a building would have used under baseline operation, relying instead on readily available zone temperature measurements. This paper's contributions could make field demonstrations of advanced HVAC control more trustworthy. Improved trustworthiness of field demonstration results could help convince decision-makers in government and business that the economic case for advanced HVAC control is attractive, which could accelerate real-world adoption and thereby reduce costs and emissions from HVAC.

\appendix
\label{proofTheorem}
\textbf{Proof of Theorem~\ref{theorem-1}.}
Subtracting the experiment dynamics~\eqref{experimentDynamics} from the baseline dynamics~\eqref{baselineDynamics} gives
\noindent\begingroup\small
\[
\begin{aligned}
C_i \frac{\text d x_i(t) }{\text d t} &= \sum_{j=0}^n \alpha_{ij} ( x_j(t) - x_i(t) ) + q_i(t) - \tilde q_i(t) ,
\end{aligned}
\] \endgroup
where $x_i(t) \coloneqq T_i(t) - \tilde T_i(t)$ is the temperature reduction in zone $i$ relative to the benchmark scenario. The cumulative cost savings in zone $i$ are therefore
\noindent\begingroup\small
\[
\begin{aligned}
S_i &\coloneqq \int_0^\tau a_i(t) ( q_i(t) - \tilde q_i(t) ) \text d t \\
&= C_i  \int_0^\tau a_i(t) \frac{\text d x_i(t) }{\text d t} \text d t + \sum_{j=0}^n \alpha_{ij}  \int_0^\tau a_i(t) ( x_i(t) - x_j(t) ) \text d t .
\end{aligned}
\] \endgroup
Applying Integration by Parts,
\noindent\begingroup\small
\[
\begin{aligned}
\int_0^\tau a_i(t) \frac{\text d x_i(t) }{\text d t} \text d t &= a_i(\tau) x_i(\tau)  - a_i(0) x_i(0)  - \int_0^\tau \frac{\text d a_i(t) }{\text d t} x_i(t) \text d t .
\end{aligned}
\] \endgroup
The building has the same initial and final states in the experiment scenario as it does in the benchmark scenario, so $x_i(0) = x_i(\tau) = 0$ and the boundary terms vanish. Hence,
\noindent\begingroup\small
\begin{equation}
\int_0^\tau a_i(t) \frac{\text d x_i(t) }{\text d t} \text d t = - \int_0^\tau \frac{\text d a_i(t) }{\text d t} x_i(t) \text d t \label{IbyP}
\end{equation} \endgroup
Substituting~\eqref{IbyP} into the expression for $S_i$ gives
\noindent\begingroup\small
\[
S_i =   - C_i \int_0^\tau \frac{\text d a_i(t) }{\text d t} x_i(t) \text d t + \sum_{j=0}^n \alpha_{ij}  \int_0^\tau a_i(t) ( x_i(t) - x_j(t) ) \text d t .
\] \endgroup
Recalling that $x_0(t) = x_{m+1}(t) = \dots = x_n(t) = 0$ for all $t$,
\noindent\begingroup\small
\[
\begin{aligned}
S_i &= \int_0^\tau \left( - C_i \frac{\text d a_i(t) }{\text d t}  + a_i(t) \sum_{j=0}^n \alpha_{ij}  \right) x_i(t) \text d t \\
&- \sum_{j=1}^m \alpha_{ij}  \int_0^\tau a_i(t) x_j(t) \text d t \\
\end{aligned}
\] \endgroup
and
\noindent\begingroup\small
\[
\begin{aligned}
S \coloneqq \sum_{i=1}^n S_i 
&= \sum_{i=1}^m \int_0^\tau \left( - C_i \frac{\text d a_i(t) }{\text d t}  + a_i(t) \sum_{j=0}^n \alpha_{ij}  \right) x_i(t) \text d t \\
&- \sum_{j=1}^n \sum_{i=1}^m \alpha_{ji}  \int_0^\tau a_j(t) x_i(t) \text d t .
\end{aligned}
\] \endgroup
Swapping the order of the sums in the second term and recalling that $\alpha_{ij} = \alpha_{ji}$ by symmetry,
\noindent\begingroup\small
\[
\begin{aligned}
S &= \sum_{i=1}^m \left[ \int_0^\tau \left( - C_i \frac{\text d a_i(t) }{\text d t}  + a_i(t) \sum_{j=0}^n \alpha_{ij}  \right) x_i(t) \text d t \right. \\
&\left. - \sum_{j=1}^n \alpha_{ij}  \int_0^\tau a_j(t) x_i(t) \text d t \right] \\
&= \sum_{i=1}^m \int_0^\tau \left[  a_i(t) \alpha_{i0} - C_i \frac{\text d a_i(t) }{\text d t} + \sum_{j=1}^n \alpha_{ij} (a_i(t) - a_j(t) ) \right] x_i(t) \text d t . \\
\end{aligned}
\] \endgroup
This exact expression for the total energy savings, along with~\eqref{IbyP}, gives the second statement of the theorem. The overestimation error from reporting savings only in the controlled zones $1$, \dots, $m$, rather than all zones, is
\noindent\begingroup\small
\[
\begin{aligned}
\sum_{i=1}^m S_i - S  &= - \sum_{i=m+1}^n S_i \\
&= - \sum_{i=m+1}^n \left[ \int_0^\tau \left( - C_i \frac{\text d a_i(t) }{\text d t}  + a_i(t) \sum_{j=0}^n \alpha_{ij}  \right) \cancel{ x_i(t) } \text d t \right. \\
&\left. - \sum_{j=1}^m \alpha_{ij}  \int_0^\tau a_i(t) x_j(t) \text d t \right] \\
&= \sum_{i=m+1}^n \sum_{j=1}^m \alpha_{ij}  \int_0^\tau a_i(t) x_j(t) \text d t \\
&= \sum_{i=1}^m \sum_{j=m+1}^n \alpha_{ij}  \int_0^\tau a_j(t) x_i(t) \text d t , 
\end{aligned}
\] \endgroup
where the cancellation of $x_i(t)$ in the second line follows because $x_{m+1}(t) = \dots = x_n(t) = 0$ for all $t \in [0, \tau]$. This exact expression for the overestimation error concludes the proof.

\section*{Acknowledgements}

The authors gratefully acknowledge support for AJK from the American Society of Heating, Refrigeration, and Air-Conditioning Engineers (ASHRAE) Graduate Student Grant-in-Aid Award and thank James E. Braun and Gregor P. Henze for helpful discussion.

{\footnotesize
\bibliographystyle{IEEEtran}
\bibliography{refs}
}

\end{document}